\documentclass[pre,10pt,twocolumn]{revtex4}%
\usepackage{amsfonts}
\usepackage{amsmath}
\usepackage{amssymb}
\usepackage{graphicx}%
\newtheorem{theorem}{Theorem}

\newtheorem{corollary}[theorem]{Corollary}

\newtheorem{lemma}[theorem]{Lemma}

\newenvironment{proof}[1][Proof]{\noindent\textbf{#1.} }{\ \rule{0.5em}{0.5em}}
\begin{document}
\title{Spin systems dynamics and faults detection in threshold networks}
\author{Steve Kirkland}
\email{stephen.kirkland@nuim.ie}
\affiliation{Hamilton Institute, National University of Ireland, Maynooth Co. Kildare, Ireland}
\author{Simone Severini}
\email{simoseve@gmail.com}
\affiliation{Department of Physics and Astronomy, University College London, WC1E 6BT
London, United Kingdom}

\begin{abstract}
We consider an agent on a fixed but arbitrary node of a known threshold
network, with the task of detecting an unknown missing link/node. We obtain
analytic formulas for the probability of success, when the agent's tool is the
free evolution of a single excitation on an $XX$ spin system paired with the
network. We completely characterize the parameters allowing for an
advantageous solution. From the results emerges an optimal (deterministic)
algorithm for quantum search, therefore gaining a quadratic speed-up with
respect to the optimal classical analogue, and in line with well-known results
in quantum computation. When attempting to detect a faulty node, the chosen
setting appears to be very fragile and the probability of success too small to
be of any direct use.

\end{abstract}
\maketitle

\section{Introduction}

Searching and traversing graphs with the use of a quantum dynamics (discrete
or continuous) is a topic of wide study. The main findings on the algorithmic
side exhibit a quadratic gain (producing then a class of Grover-like results),
when used to search marked vertices in hypercubes, lattices, and more general
objects \cite{se}, and an exponential one when transferring information to a
distant site by free evolution \cite{ex}. With essentially the same formalism,
the dynamics of spin systems forms the basis of concrete proposals for
implementing communication buses in a variety of nanotechnology devices
\cite{so}. The central point is always a unitary operator reflecting the
topology of a graph whose vertices are encoded in pure states or a
(time-independent) Hamiltonian describing the interactions of particles in the
actual physical network, and thus representing a noiseless quantum channel.

Considering these two contexts together, we take a parsimonious agent located
on a specific node of a known threshold network, attempting to detect an
unknown missing link/node. We obtain analytic formulas for the probability of
success, when the missing link/node is searched by letting evolve a single
excitation from the given vertex. For specific parameters, we have an optimal
algorithm for quantum search in a deterministic fashion. Our working ground is
an \emph{XX} (or, equivalently, isotropic \emph{XY}) system with homogeneous
couplings and site dependent magnetic fields \cite{cas} (a proposed setting
for distributed implementation consists of cavities connected by optical
fibers \cite{im}).

Threshold graphs are applied to the synchronization of parallel processes, set
packing, scheduling, and to define the Guttman scale in the area of
statistical surveys \cite{re}. Formally, a \emph{threshold graph} can be
constructed from the one-vertex graph by repeatedly adding a single vertex of
two possible types: an \emph{isolated vertex}, \emph{i.e.} a vertex without
edges; a \emph{dominating vertex}, \emph{i.e.} a vertex connected to all other
vertices. The \emph{creation sequence} of a threshold graph $G$ on $n$
vertices is a word $x\left(  G\right)  \in\{0,1\}^{n}$ recursively defined as
follows: $x\left(  G\right)  =x\left(  G-i\right)  y$, where $y=0$ if $i$ is
isolated and $y=1$ if $i$ is dominating. A threshold graph is then
characterized by its creation sequence, a property with important
consequences. In fact, the information required to store a threshold graph on
$n$ vertices is at most $n$ bits, two threshold graphs are isomorphic if and
only if they have the same degree sequence, and the recognition can be done in
linear time. Even if the above definition is somehow restrictive, it is
valuable that one can construct arbitrarily large threshold networks with
approximately any prescribed degree distribution, including the scale-free one
\cite{ha}.

The dynamics of the first excitation sector in the \emph{XX} model is governed
by the (combinatorial) Laplacian \cite{cas}. The \emph{Laplacian} of a graph
$G=\left(  V,E\right)  $ is the matrix $L$ with entries $[L]_{i,j}=d\left(
i\right)  $ if $i=j$, $[L]_{i,j}=-1$ if $\{i,j\}\in E$ and $[L]_{i,j}=0$ if
$\{i,j\}\notin E$; here $d\left(  i\right)  =\left\vert \left\{  j:\{i,j\}\in
E\right\}  \right\vert $ is the \emph{degree} of $i\in V$. Let $\lambda
_{1}\geq\lambda_{2}\geq\cdots\geq\lambda_{n}=0$ be the eigenvalues of $L$
arranged in the nonincreasing order. A graph is a threshold graph if and only
if $\lambda_{j}=\left\vert \left\{  i:d\left(  i\right)  \geq j\right\}
\right\vert $, with $j=1,...,n$ \cite{me}. In other words, the degree sequence
and the nonzero eigenvalues of a threshold graph $G$ are (Ferrer's) conjugate
partitions of $2\left\vert E\right\vert $. For example, let us consider the
word $0011011$. The graph $G$ has degree sequence $\{2,4,4,5,5,6,6\}$ and
eigenvalues $\{7,7,6,6,4,2,0\}$.

Let $|\psi_{t}\rangle=U_{t}|\psi_{0}\rangle$, where $U_{t}\equiv\exp(-iLt)$,
for $t\in\mathbb{R}^{+}$, and $|\psi_{0}\rangle$ is an element of an $n$
dimensional Hilbert space, with standard basis $\{|i\rangle:i\in V\}$. Let
$p_{G}\left(  i,j,t\right)  =\left\vert \langle j|U_{t}|i\rangle\right\vert
^{2}$ be the probability that a single excitation travels from node $i$ to
node $j$ after a free evolution of duration $t$ (this is also called
\emph{fidelity}); we say that there is \emph{perfect state transfer} when
$p_{G}\left(  i,j,t\right)  =1$, for $i\neq j$. Since the Laplacian spectrum
of a threshold graph is integral, it follows that the dynamics governed by $L$
is always periodic \cite{so, god}, \emph{i.e.}, there is $t$ such that
$p_{G}\left(  i,i,t\right)  =1$, for every $i$.

Firstly, we propose in Theorem \ref{com} a complete characterization of the
class of threshold graphs allowing for perfect state transfer. Secondly, with
a direct application of this result, we describe how the evolution can be used
to detect a missing link. The procedure gives a straightforward algorithm for
(optimal) quantum search in certain threshold networks. Indeed, as a special
case, we have a deterministic version of Grover's search \cite{be}. Theorem
\ref{com2} shows that the result can not be extended to a large subfamily of
threshold graphs. As expected, the graphs for which the result holds can be
associated to modulated chains of length $3$, where the couplings are
determined by the number of edges connecting vertices from certain subsets
(namely, an equitable partition). This is a common situation when studying
perfect state transfer \cite{so}. Interference effects drive the evolution in
a way that does not seem to be exploitable for searching (by using out
technique), unless the missing links belong to the subsets.

\section{Faulty links}

For any $p\in\mathbb{N},$ let $K_{p}$ denote the complete graph on $p$
vertices, and let $O_{p}$ denote the empty graph on $p$ vertices. Let
$\mathbf{0}_{p}$ and $\mathbf{1}_{p}$ denote the zero vector and all ones
vector, respectively, both of order $p$. For $p,q\in\mathbb{N},$ we use
$J_{p\times q}$ to denote the $p\times q$ all ones matrix; often the subscript
will be suppressed when the order is evident from the context. We need two
operations: the union and the joint. The \emph{union}, denoted by $\cup$,
consists of taking two graphs and looking at them as a single one whose
connected components are exactly the graphs. The \emph{join}, denoted by
$\vee$, is the graph obtained by taking the union of two graphs, plus edges
connecting all their respective vertices. The following lemma can be deduced
from the basic properties of threshold graphs and it does not require a proof
(see, \emph{e.g.}, \cite{me}). It is our main technical tool.

\begin{lemma}
\label{lem}Let $G$ be a connected graph on at least two vertices. Then $G$ is
a threshold graph if and only if one of the following two conditions is satisfied:

\begin{enumerate}
\item[(i)] there are indices $m_{1},\ldots,m_{2k}\in\mathbb{N}$ with
$m_{1}\geq2$ such that $G=((((O_{m_{1}}\vee K_{m_{2}})\cup O_{m_{3}})\vee
K_{m_{4}})\ldots)\vee K_{m_{2k}}\equiv\Gamma(m_{1},\ldots,m_{2k})$;

\item[(ii)] there are indices $m_{1},\ldots,m_{2k+1}\in\mathbb{N}$ with
$m_{1}\geq2$ such that $G=((((K_{m_{1}}\cup O_{m_{2}})\vee K_{m_{3}})\cup
O_{m_{4}})\ldots)\vee K_{m_{2k+1}}\equiv\Gamma(m_{1},\ldots,m_{2k+1})$.

The vertices $1.,..,m_{1}$ correspond to the first subset, $m_{1}%
+1,...,m_{1}+m_{2}$ correspond to the second subset, \emph{etc.}
\end{enumerate}
\end{lemma}

The next theorem is the central result of this section. It gives a complete
characterization of threshold graphs with perfect state transfer. As it is
usually done when studying this topic, we write explicitly the eigensystem of
the Laplacian. The parameterization is the one described in Lemma \ref{lem}.

\begin{theorem}
\label{com}Let $G$ be a threshold graph. When $G\equiv\Gamma(m_{1}%
,\ldots,m_{2k})$ (resp. $G\equiv\Gamma(m_{1},\ldots,m_{2k+1})$), $p_{G}\left(
i,j,t\right)  =1$ if and only if $\left(  i,j\right)  =\left(  1,2\right)  $
and in addition: $t=\frac{\pi}{2}$; $m_{1}=2$; $m_{2}\equiv2\operatorname{mod}%
4$; and $m_{j}\equiv0\operatorname{mod}4,j=3,\ldots,2k$ (resp. $j=3,\ldots
,2k+1$).
%\item[(ii)] $t=\frac{3\pi}{2}$, $m_{1}=2,m_{2}\equiv2\operatorname{mod}4$, and
%$m_{j}\equiv0\operatorname{mod}12$, $j=3,\ldots,2k$ (resp. $j=3,\ldots,2k+1$).

\end{theorem}

\begin{proof}
Suppose that we have $m_{1},\ldots,m_{2k}\in\mathbb{N}$ with $m_{1}\geq2$, and
consider the Laplacian matrix $L$ of $\Gamma(m_{1},\ldots,m_{2k})$. For each
$1\leq l\leq2k$, let $\sigma_{l}=\sum_{p=1}^{l}m_{p}$. Fix an odd index $j$
between $1$ and $2k$, and note that if $m_{j}\geq2$, then for any vector
$u\in\mathbb{R}^{m_{j}}$ such that $u\perp\mathbf{1}_{m_{j}}$, the vector
$\left[  \mathbf{0}_{\sigma_{j-1}}|u|\mathbf{0}_{\sigma_{2k}-\sigma_{j}%
}\right]  ^{T}$ is an eigenvector for $L$ corresponding to the eigenvalue
$\lambda_{0}(j)\equiv m_{j+1}+m_{j+3}+\ldots+m_{2k}$. Letting $u_{1}%
,\ldots,u_{m_{j}-1}$ be an orthonormal basis of $(\mathbf{1}_{m_{j}})^{\perp}%
$, we find that $\sum_{l=1}^{m_{j}-1}u_{l}u_{l}^{T}=I-\frac{1}{m_{j}}%
J_{m_{j}\times m_{j}}$. Note also that if $j$ is odd and $2\leq j\leq2k$, then
the vector
\begin{equation}
\overrightarrow{\lambda}_{0}(j):=\left[
\begin{array}
[c]{c}%
\left(  \frac{m_{j}}{\sigma_{j-1}\sigma_{j}}\right)  ^{1/2}\mathbf{1}%
_{\sigma_{j-1}}\\\hline
-\left(  \frac{\sigma_{j-1}}{m_{j}\sigma_{j}}\right)  ^{1/2}\mathbf{1}_{m_{j}%
}\\\hline
\mathbf{0}_{\sigma_{2k}-\sigma_{j}}%
\end{array}
\right]  \label{eig}%
\end{equation}
is also an eigenvector for $L$ corresponding to $\lambda_{0}(j)$ that has
length one and is orthogonal to the eigenvectors constructed above. Similarly,
if $j$ is an even index between $2$ and $2k$, and $u_{1},\ldots,u_{m_{j}-1}$
is an orthonormal basis of $(\mathbf{1}_{m_{j}})^{\perp}$, then the vectors
$\left[  \mathbf{0}_{\sigma_{j-1}}|u_{l}|\mathbf{0}_{\sigma_{2k}-\sigma_{j}%
}\right]  ^{T}$, with $l=1,\ldots m_{j}-1$, are eigenvectors for $L$
corresponding to the eigenvalue $\lambda_{1}(j)\equiv\sigma_{j}+m_{j+2}%
+m_{j+4}+\ldots+m_{2k}$. Note further that the vector in Eq. (\ref{eig}) is an
eigenvector for $L$ corresponding to $\lambda_{1}(j)$.

Finally, we observe that $\mathbf{1}_{\sigma_{2k}}/\sqrt{\sigma_{2k}}$ is a
null vector for $L$. It now follows that for each index $j$ between $2$ and
$2k$, the matrix
\begin{multline*}
P[\lambda_{x}(j)]=\left[
\begin{array}
[c]{c|c|c}%
\mathbf{0} & \mathbf{0} & \mathbf{0}\\\hline
\mathbf{0} & I-\frac{1}{m_{j}}J & \mathbf{0}\\\hline
\mathbf{0} & \mathbf{0} & \mathbf{0}%
\end{array}
\right] \\
+\overrightarrow{\lambda}_{0}(j)\left[
\begin{array}
[c]{c|c|c}%
\sqrt{\frac{m_{j}}{\sigma_{j-1}\sigma_{j}}}\mathbf{1}_{\sigma_{j-1}}^{T} &
-\sqrt{\frac{\sigma_{j-1}}{m_{j}\sigma_{j}}}\mathbf{1}_{m_{j}}^{T} &
\mathbf{0}_{\sigma_{2k}-\sigma_{j}}^{T}%
\end{array}
\right] \\
=\left[
\begin{array}
[c]{c|c|c}%
\frac{m_{j}}{\sigma_{j-1}\sigma_{j}}J_{\sigma_{j-1}\times\sigma_{j-1}} &
-\frac{1}{\sigma_{j}}J & \mathbf{0}\\\hline
-\frac{1}{\sigma_{j}}J & I-\frac{1}{\sigma_{j}}J_{m_{j}\times m_{j}} &
\mathbf{0}\\\hline
\mathbf{0} & \mathbf{0} & \mathbf{0}%
\end{array}
\right]  ,
\end{multline*}
with $x=0,1$, is an orthogonal idempotent (\emph{i.e.}, a projector) for $L$
corresponding to the eigenvalue $\lambda_{0}(j)$ or $\lambda_{1}(j)$ according
as $j$ is odd or even, respectively. Here the $(1,1)$ block is of order
$\sigma_{j-1},$ the $\left(  2,2\right)  $ block is of order $m_{j},$ and
$\mathbf{0}$ denotes a zero block whose order is determined from the context.
It now follows that for any real $t\geq0$, we can use the orthogonal
idempotents to write%
\begin{align}
U_{t}  &  =e^{-it\lambda_{0}(1)}\left[
\begin{array}
[c]{c|c}%
I-\frac{1}{m_{1}}J_{m_{1}\times m_{1}} & 0\\\hline
0 & 0
\end{array}
\right] \label{exp}\\
&  +\sum_{j\geq3,j\text{{, odd}}}e^{-it\lambda_{0}(j)}P[\lambda_{0}%
(j)]\nonumber\\
&  +\sum_{j\geq2,j\text{{, even}}}e^{-it\lambda_{1}(j)}P[\lambda_{i}%
(j)]+\frac{1}{\sigma_{2k}}J.\nonumber
\end{align}
(We note that an equivalent expression for $U_{t}$ appears in Theorem 3.1 of
\cite{ik}.)

Now we consider an off-diagonal entry $z$ of $U_{t}$; for concreteness, we
take $z$ to be in the upper triangle of the matrix, and we suppose that $z$ is
in a column that lies in the $j_{0}$-th subset of the partitioning that is
induced by the indices $m_{1},\ldots,m_{2k}$. (We observe in passing that all
such offdiagonal entries are equal.) From Eq. (\ref{exp}), we find that%
\[
z=e^{-it\lambda_{p_{j_{0}}}(j_{0})}(\frac{-1}{\sigma_{j_{0}}})+\sum
_{j=j_{0}+1}^{2k}e^{-it\lambda_{p_{j}}(j)}\frac{m_{j}}{\sigma_{j-1}\sigma_{j}%
}+\frac{1}{\sigma_{2k}},
\]
where for each index $j$, $p_{j}$ is $0$ or $1$ according as $j$ is odd or
even, respectively. Consequently,
\[
\left\vert z\right\vert \le\frac{1}{\sigma_{j_{0}}}+\sum_{j=j_{0}+1}^{2k}%
\frac{m_{j}}{\sigma_{j-1}\sigma_{j}}+\frac{1}{\sigma_{2k}}.
\]
Note that $m_{2k}/(\sigma_{2k-1}\sigma_{2k})+1/\sigma_{2k}=1/\sigma_{2k-1}$.
Thus, the summation giving $\left\vert z\right\vert $ telescopes, so that%
\[
\sum_{j=j_{0}+1}^{2k}\frac{m_{j}}{\sigma_{j-1}\sigma_{j}}+\frac{1}{\sigma
_{2k}}=\frac{1}{\sigma_{j_{0}}}.
\]
Hence, $\left\vert z\right\vert \leq2/\sigma_{j_{0}}$. In particular, since
$m_{1}\geq2$, we see that $\left\vert z\right\vert <1$ if $j_{0}\geq2$.

Suppose now that $\left\vert z\right\vert =1$. Then necessarily $j_{0}=1,$
and, since $1\leq2/m_{1}$, we find that $m_{1}=2$. Indeed, it must also be the
case that $z$ is the $(1,2)$ entry of $U_{t}$. From the consideration of the
equality case in the triangle inequality, we find that $\left\vert
z\right\vert =1$ if and only if, in addition to $m_{1}=2$ we have:

\begin{enumerate}
\item[(i)] \emph{ }$e^{-it\lambda_{0}(1)}=-1$, while

\item[(ii)] $e^{-it\lambda_{0}(j)}=1$, for all odd $j\geq3$ and

\item[(iii)] $e^{-it\lambda_{1}(j)}=1$, for all even $j$.
\end{enumerate}

From $\emph{(i)}$ we have $e^{-2ti}=-1$, while from \emph{(ii)} it follows
that $e^{-itm_{2l}}=1$ for each $l=2,\ldots,k$. Applying these conditions, in
conjunction with \emph{(iii)}, it then follows that $e^{-itm_{2}}=-1$, while
$e^{-itm_{2l+1}}=1$, for $l=1,\ldots,k-1$. Hence there are integers
$p_{1},\ldots,p_{2k}$ such that%
\[
\frac{t}{\pi}=\frac{2p_{1}+1}{2}=\frac{2p_{2}+1}{m_{2}}=\frac{2p_{j}}{m_{j}},
\]
for $j=3,\ldots,2k$. We thus find that $m_{2}=\left(  4p_{2}+2\right)
/\left(  2p_{1}+1\right)  $ while $m_{j}=4p_{j}/\left(  2p_{1}+1\right)  $,
$j=3,\ldots,2k$. We now deduce that $m_{2}\equiv2\operatorname{mod}4$, and
$m_{j}\equiv0\operatorname{mod}4$ for $j=3,\ldots,2k$.

Next, we consider the values of $t\in\lbrack0,2\pi]$ for which the $(1,2)$
entry of $U_{t}$ has modulus one, and note that since $L$ has integer
eigenvalues, an offdiagonal entry of $U_{t}$ has modulus one if and only if
the corresponding entry in $\exp(-i(t+2q\pi)L)$ has modulus one for any
integer $q$. From the foregoing we find that the only possible values of
$t\in\lbrack0,2\pi]$ for which the $(1,2)$ entry of $U_{t}$ has modulus one
are $t=\pi/2$ and $t=3\pi/2$. Since $m_{2}\equiv2\operatorname{mod}4$ and
$m_{j}\equiv0\operatorname{mod}4,$ for $j=3,\ldots,2k,$ we find readily that
conditions i)-iii) hold when $t=\pi/2$ or $t=3\pi/2,$ so that the $(1,2)$
entries of $\exp(-i\pi L/2)$ and $\exp(-i3\pi L/2)$ have modulus one. The
second part of the theorem follows along the same lines.
\end{proof}

\begin{corollary}
Let $G$ be a threshold graph on $n=4m$ vertices with $m_{1}=2$ and
$m_{2}=4m-2$. The pair $\{i,j\}$, with $i$ and $j$ in the set of size $m_{1}$,
can be found with certainty by the use of $O(n-1)$ evolutions induced by $L$,
each one of time $\pi/2$.
\end{corollary}

\begin{proof}
Consider $K_{n},$ the complete graph on $n$ vertices. A special case of
Theorem \ref{com} is $K_{n}^{-}:=K_{n}-\left\{  i,j\right\}  $, for two
distinct arbitrary vertices $i$ and $j$. When $n=4m$ ($m\geq1$), $p_{K_{n}%
^{-}}\left(  i,j,\pi/2\right)  =1$ and $p_{K_{n}^{-}}\left(  k,l,\pi/2\right)
=0$, for every pair $\{k,l\}$, with $k,l\neq i,j$. Since $|E(K_{n}%
)|\rightarrow n^{2}$ for $n\rightarrow\infty$, $n-1$ applications of
$U_{\pi/2}$ implement a deterministic and optimal version of Grover's search
\cite{be}.
\end{proof}

\bigskip

In practice, we are given a network modeled by a complete $K_{n}$ that is
missing a single unknown link $\{i,j\}$. By starting the dynamics on each
possible vertex (or, equivalently, particle) of $K_{n}^{-}$, we can determine
the missing edge by letting the entire system evolve for a time $\pi/2$ and
then perform a projective measurement at the same vertex. The dynamics governs
the behaviour of a walker on the network. If the walker has moved, then the
new position is vertex $j$ and the link $\{i,j\}$ is missing.

The same reasoning may be applied to find a missing matching. Recall that a
\emph{matching} is a set of vertex-disjoint edges. A matching is
\emph{perfect} if it includes all vertices.

\begin{corollary}
Let $K_{n}$ be a graph on $n=4m$ vertices. A deleted matching from $K_{n}$ of
size $k\leq2m$ can be found with certainty by the use of exactly $n/2-1$
evolutions induced by $L$, each one of time $\pi/2$.
\end{corollary}

It is just matter of running the dynamics from an arbitrary vertex $i$ and
detecting the missing edge $\{i,j\}$ in the matching. Then, we pass to a
vertex $k\neq i,j$, and so on and so forth. A deterministic search requires at
most $%
%TCIMACRO{\tsum \nolimits_{i=3\text{: odd}}^{n-1}}%
%BeginExpansion
{\textstyle\sum\nolimits_{i=3\text{: odd}}^{n-1}}
%EndExpansion
i=n^{2}/4-1$ steps. For threshold graphs, Theorem \ref{com} specifies
completely the detectable links. From the perspective of a direct application
to network search based on the described method (free-evolution on a spin
system), the theorem shows that the complete graph is the only threshold graph
in which every deleted link can be actually found. For all other threshold
graphs there are some undetectable links. Even if the same procedure can be
certainly generalized to any graph, since the search is performed in cliques
(\emph{i.e.}, complete subgraphs), a complete knowledge of cliques is
necessary and a mechanism to induce evolution only in desired cliques is
needed. Still, it is useful to remark that implementing such a mechanism
(which would turn on and off the interactions between different cliques)
permits the transfer of an excitation to any desired vertex.

\section{Faulty nodes}

Here we consider the problem of the previous section but for vertices. Given a
threshold graph $G$, let $\hat{G}=G-j$, for some vertex $j\in V$, and let
$\hat{U}_{t}=\exp(-iL(\hat{G})t)$. We ask whether it is plausible that if $G$
has perfect state transfer, we are then able to detect the missing vertex. The
idea is based on taking advantage of the graph structure, apart from the edge
between the two vertices in the set of size $m_{1}$. After a technical lemma,
we will prove that the offdiagonal entries of $\hat{U}_{t}$, while bounded
away from $1$ in modulus, can have relatively large magnitude, something which
does not help to give specific information about the deleted vertex. As a
consequence, we do not obtain sufficient information about $j$. The negative
result highlights an important role for special symmetries. We shall give a
proof for $G=\Gamma(m_{1},m_{2},\ldots,m_{2k})$. The theorem for the case
$\Gamma(m_{1},m_{2},\ldots,m_{2k+1})$ has a parallel statement.

\begin{lemma}
\label{cos}Let $a\in\mathbb{N\,}$\ be odd. Then

\begin{enumerate}
\item[(i)] $\max\{\min\{-\cos2t,-\cos at,\cos(a-2)t\}|t\in\lbrack0,2\pi]\}$

$=\cos(\pi/a)$;

\item[(ii)] $\max\{\min\{-\cos2t,\cos at,-\cos(a-2)t\}|t\in\lbrack0,2\pi]\}$

$=\cos(\pi/a)$.
\end{enumerate}
\end{lemma}

\begin{proof}
Set $a=2m+1$. If%
\[
t\in S=\left[  \frac{m}{2m+1}\pi,\frac{m+1}{2m+1}\pi\right]  \cup\left[
\frac{3m+1}{2m+1}\pi,\frac{3m+2}{2m+1}\pi\right]  ,
\]
then $-\cos2t<\cos\left(  \pi/a\right)  $. On the other hand, if $t\notin S$
then $\left\vert \cos(a-2)t\right\vert \leq\cos\left(  \pi/a\right)  $. Thus,
we find that%
\[%
\begin{tabular}
[c]{l}%
$\min\{-\cos2t,-\cos at,\cos(a-2)t\},$\\
$\min\{-\cos2t,\cos at,-\cos(a-2)t\}\leq\cos\dfrac{\pi}{a},$%
\end{tabular}
\]
for all $t\in\lbrack0,2\pi]$. Let $b=\left(  m+1\right)  /\left(  2m+1\right)
$ and $c=m/\left(  2m+1\right)  $. Next we show that for each of the functions
above, the value $\cos(\pi/a)$ is attainable for some $t$. If $m$ is odd/even
then
\begin{align*}
\cos(\pi/a) =  &  \left\{  -\cos2(c\pi),-\cos a(c\pi),\cos(a-2)(c\pi)\right\}
,\\
&  \left\{  -\cos2(b\pi),-\cos a(b\pi),\cos(a-2)(b\pi)\right\}
\end{align*}
respectively; in a similar way, if $m$ is even/odd%
\begin{align*}
\cos(\pi/a) =  &  \left\{  -\cos2(c\pi),\cos a(c\pi),-\cos(a-2)(c\pi)\right\}
,\\
&  \left\{  -\cos2(b\pi),\cos a(b\pi),-\cos(a-2)(b\pi)\right\}  .
\end{align*}

\end{proof}

\begin{theorem}
\label{com2}Let $m_{1},\ldots,m_{2k}\in\mathbb{N}$ such that $m_{1}=2$,
$m_{2}\equiv2\operatorname{mod}4$, and $m_{l}\equiv0\operatorname{mod}4$, for
$l=3,\ldots,2k$. Let $i\neq j$ and $g=\left(  \frac{1}{2}+\frac{m_{2}}%
{\sigma_{1}\sigma_{2}}+\frac{1}{\sigma_{2k}}\right)  ^{2}$.

\begin{enumerate}
\item[(i)] If $\hat{G}=\Gamma(m_{1}-1,m_{2},\ldots,m_{2k})$ then $|[\hat
{U}_{t}]_{i,j}|\leq2/\left(  m_{2}+1\right)  $;

\item[(ii)] If $l$ is even and $\hat{G}=\Gamma(m_{1},\ldots,m_{l-1}%
,m_{l}-1,m_{l+1},\ldots,m_{2k})$ then%
\begin{align}
|[\hat{U}_{t}]_{i,j}|  &  \leq\left[  g\right.  -\left(  1-\cos\frac{\pi}{1+%
%TCIMACRO{\tsum \nolimits_{i=2\text{: even}}^{2k}}%
%BeginExpansion
{\textstyle\sum\nolimits_{i=2\text{: even}}^{2k}}
%EndExpansion
m_{i}}\right) \nonumber\\
&  \left.  \frac{m_{2}}{\sigma_{1}\sigma_{2}\sigma_{2k}}\right]  ^{1/2}%
+\frac{1}{2}-\frac{m_{2}}{\sigma_{1}\sigma_{2}}-\frac{1}{\sigma_{2k}};
\label{eq1}%
\end{align}

\item[(iii)] If $l\geq3$ is odd and $\hat{G}=\Gamma(m_{1},\ldots,m_{l-1}%
,m_{l}-1,m_{l+1},\ldots,m_{2k})$ then
\begin{align}
|[\hat{U}_{t}]_{i,j}|  &  \leq\left[  g\right.  -\left(  1-\cos\frac{\pi}%
{\sum_{i=1\text{: odd}}^{l}m_{i}-1}\right) \nonumber\\
&  \left.  \frac{m_{2}m_{l+1}}{\sigma_{1}\sigma_{2}\sigma_{l}\sigma_{l+1}%
}\right]  ^{1/2}+\frac{1}{2}-\frac{m_{2}}{\sigma_{1}\sigma_{2}}-\frac{m_{l+1}%
}{\sigma_{l}\sigma_{l+1}}. \label{eq2}%
\end{align}

\end{enumerate}
\end{theorem}

\begin{proof}
Let $z$ be an offdiagonal entry of $\hat{U}_{t}$. As in the proof of Theorem
\ref{com}, we find that if $z\neq\lbrack\hat{U}_{t}]_{1,2},[\hat{U}_{t}%
]_{2,1}$ then $\left\vert z\right\vert \leq2/\sigma_{2}\leq2/3$. Thus, for the
remainder of the proof, we assume \emph{wlog} that $z$ is in the $(1,2)$
position. Recall that $\sigma_{l}=\sum_{p=1}^{l}m_{p}$.

\emph{(i) }Note that $\Gamma(m_{1}-1,m_{2},\ldots,m_{2k})=\Gamma(m_{2}%
+1,m_{3},\ldots,m_{2k})$. From the proof of Theorem \ref{com}, it follows that
$\left\vert z\right\vert \leq2/\left(  m_{2}+1\right)  $.

\emph{(ii) }Suppose that $l$ is even and that $\hat{G}=\Gamma(m_{1}%
,\ldots,m_{l-1},m_{l}-1,m_{l+1},\ldots,m_{2k})$. Set $a=1+m_{2}+m_{4}%
+\ldots+m_{2k}$. Suppose that we have positive parameters $\alpha\beta,\gamma$
such that $\alpha\geq\beta,\gamma$, and note that $\left\vert -\alpha
e^{-it(a-2)}+\beta e^{-ita}+\gamma\right\vert ^{2}=\alpha^{2}+\beta^{2}%
+\gamma^{2}-2\alpha\gamma\cos(a-2)t+2\beta\gamma\cos at-2\alpha\beta\cos2t$.
From Lemma \ref{cos}, for each $t$, one of $-\cos2t$, $\cos at$ and
$-\cos(a-2)t$ is bounded above by $\cos\left(  \pi/a\right)  $. It now follows
that $\left\vert -\alpha e^{-it(a-2)}+\beta e^{-ita}+\gamma\right\vert
\leq(\alpha^{2}+\beta^{2}+\gamma^{2}+2\alpha\gamma+2\beta\gamma\cos\frac{\pi
}{a}+2\alpha\beta)^{1/2}=((\alpha+\beta+\gamma)^{2}-(1-\cos\frac{\pi}{a}%
)\beta\gamma)^{1/2}$. Taking $\alpha=1/\sigma_{1}=1/2$, $\beta=m_{2}%
/(\sigma_{1}\sigma_{2})$, and $\gamma=1/\sigma_{2k}$, applying the triangle
inequality as in the proof of Theorem \ref{com}, we have
\begin{align*}
\left\vert z\right\vert  &  \leq\left\vert \frac{-1}{2}e^{-it(a-2)}%
+\frac{m_{2}}{\sigma_{1}\sigma_{2}}e^{-ita}+\frac{1}{\sigma_{2k}}\right\vert
+\left\vert \sum_{j=3}^{2k}\frac{m_{j}}{\sigma_{j-1}\sigma_{j}}\right\vert \\
&  \leq\left\vert \frac{1}{2}e^{-it(a-2)}+\frac{m_{2}}{\sigma_{1}\sigma_{2}%
}e^{-ita}+\frac{1}{\sigma_{2k}}\right\vert +1\\
&  -\frac{1}{2}-\frac{m_{2}}{\sigma_{1}\sigma_{2}}-\frac{1}{\sigma_{2k}},
\end{align*}
which implies Eq. (\ref{eq1}).

\emph{(iii) }Suppose that $l\geq3$ is odd, and that $\hat{G}=\Gamma
(m_{1},\ldots,m_{l-1},m_{l}-1,m_{l+1},\ldots,m_{2k})$. Set $a=m_{1}%
+m_{3}+\ldots+m_{l}-1$. As in the proof of \emph{(ii)}, we consider positive
parameters $\alpha\geq\beta,\gamma$, and note that $\left\vert -\alpha
e^{-it\lambda_{0}(1)}+\beta e^{-it\lambda_{1}(2)}+\gamma e^{-it\lambda
_{1}(l+1)}\right\vert =\left\vert -\alpha+\beta e^{-2it}+\gamma e^{-iat}%
\right\vert $. As in \emph{(ii)}, we deduce from Lemma \ref{cos} that
$\left\vert -\alpha+\beta e^{-2it}+\gamma e^{-iat}\right\vert \leq
((\alpha+\beta+\gamma)^{2}-(1-\cos\frac{\pi}{a})\beta\gamma)^{1/2}$. Taking
$\alpha=1/\sigma_{1}=1/2$, $\beta=m_{2}/(\sigma_{1}\sigma_{2})$, and
$\gamma=m_{l+1}/(\sigma_{l}\sigma_{l+1})$, Eq. (\ref{eq2}) follows.
\end{proof}

\bigskip

Notice that in the context of Theorems \ref{com2}, the $(1,2)$ entry of
$\hat{U}_{t}$ can have large modulus. For example, it is straightforward to
show that if $m_{1}=2$, $m_{2}\equiv2\operatorname{mod}4$, and $m_{l}%
\equiv0\operatorname{mod}4$, for $l=3,\ldots,2k$ and $\hat{G}=\Gamma
(m_{1},m_{2},\ldots,m_{2k-1},m_{2k}-1)$, then $|[\hat{U}_{\frac{\pi}{2}%
}]_{1,2}|=\left(  1-2\left(  (\sigma_{2k}-1)/\sigma_{2k}^{2}\right)  \right)
^{1/2}$. Similarly, if $m_{1}=2$, $m_{2}\equiv2\operatorname{mod}4$, and
$m_{l}\equiv0\operatorname{mod}4$, for $l=3,\ldots,2k+1$ and $\hat{G}%
=\Gamma(m_{1},m_{2},\ldots,m_{2k},m_{2k+1}-1)$, then $|[\hat{U}_{\frac{\pi}%
{2}}]_{1,2}|=\left(  1-2\left(  (\sigma_{2k+1}-1)/\sigma_{2k+1}^{2}\right)
\right)  ^{1/2}$.

\section{Conclusions}

We have shown that the dynamics on a network governed by the Laplacian, seen
as the Hamiltonian restricted to the single excitation sector of an \emph{XX}
spin system, can detect and find a faulty link in the complete graph with a
quadratic gain with respect to a deterministic method, and we have extended
the observation to matchings. The result gives a way to perform optimal
quantum search (on the complete graph), and a new insight into the related
algorithms. Essentially we have a reinterpretation of a fact discovered in
\cite{cas}. Our contribution is to have put the statement in a more general
mathematical context, by giving a complete characterized of threshold graphs
with perfect state transfer with respect to the \emph{XX} model. Dealing with
deleted nodes, we have shown that our method does not give any clear
advantage. The basic idea of the paper is to use a free quantum evolution to
search a missing item on a network. The method is different from the ones
studies in \cite{se} substantially because the algorithm does not require any
control. After the set-up of the network, the system evolves without
intermediate operations. The process is distributed because the nodes of the
network are identified with spin particles. The measurements are local in the
sense that they are concerned with the single sites, independently. That is
why we can locate autonomous agents on the sites.

We have considered threshold graphs only, since these have integer Laplacian
spectrum, and so possess periodic dynamics, \emph{i.e.}, a necessary condition
for perfect state transfer (see \cite{so}). The vertices involved in the
phenomenon are special. The symmetry with respect to these vertices can be
exploited to create a \textquotedblleft reference point\textquotedblright%
\ inside the graph. Looking ahead, one direction for further exploration is to
determine what information about the topology of a spin system (paired with a
network)\ can be obtained by a free evolution and final local measurements.
Here, more than designing search algorithms, it is a matter to determine what
kind of graphs have some sort of searching capability embedded in their
structure. The method of the paper can be generalized to searching a missing
link (or, equivalently, a marked link) in any network, with steps involving
one clique at a time. The method works for vertices if we employ the notion of
the line graph. However, we have shown that the direct detection of a missing
vertex is not a natural task for the studied dynamics, at least on threshold
graphs. It is an open question to determine whether the dynamics can help to
find marked nodes, when taking a different initial probability distribution,
and if we can obtain the quadratic speed-up in this case.

\bigskip

\noindent\emph{Acknowledgments.} We would like to thank Sougato Bose, Abolfazl
Bayat, Daniel Burgarth, Domenico D'Alessandro, Aram Harrow, and Vladimir
Korepin, for helpful discussion. We would also like to thank one of the
anonymous referees for important comments that contributed to improve the
paper. SK is supported in part by the Science Foundation Ireland under Grant
No. SFI/07/SK/I1216b. SS is supported by a Newton International Fellowship.

\end{document}